\documentclass[a4paper,12pt]{article}
\usepackage{amsthm}
\usepackage{amsmath}
\usepackage{amssymb}
\usepackage{graphicx}
\usepackage{newalg}
\usepackage{a4wide}

\DeclareMathOperator{\Wdeg}{wdeg}
\DeclareMathOperator{\F}{\mathbb{F}}

\DeclareMathOperator{\wdeg}{\Wdeg_{\mathit{(1,k-1)}}}

\DeclareMathOperator{\lt}{LT}

\DeclareMathOperator{\ydeg}{ydeg}

\theoremstyle{definition}

\theoremstyle{plain}
\newtheorem{lemma}{Lemma}
\newtheorem{thm}{Theorem}
\theoremstyle{remark}
\newtheorem{rem}{Remark}

\begin{document}
\title{Efficient Interpolation in the Guruswami-Sudan Algorithm}
\author{P.V. Trifonov
%,~\IEEEmembership{Member,~IEEE}
\thanks{P.V. Trifonov is with the Distributed Computing and Networking Department,
Saint-Petersburg State Polytechnic University, Russia. e-mail:ptrifonov@ieee.org}}
%\author{P.V. Trifonov}

\markboth{IEEE Transactions on Information Theory, submitted for publication, December 2008}
{P.V. Trifonov. Efficient Interpolation in the Guruswami-Sudan Algorithm}
\maketitle
\begin{abstract}
A novel algorithm is proposed for the interpolation step of the Guruswami-Sudan list decoding algorithm.
The proposed method is based on the binary exponentiation algorithm, and can be considered as an 
extension of the Lee-O'Sullivan algorithm.  The algorithm is shown to achieve
both asymptotical and practical performance gain compared to the case of iterative interpolation algorithm. 
Further complexity reduction is achieved by integrating the proposed method with re-encoding. 
The key contribution of the paper, which enables the complexity reduction,
is a novel randomized ideal multiplication algorithm.
\end{abstract}
%\begin{IEEEkeywords}
%List decoding, Reed-Solomon codes, interpolation, Guruswami-Sudan algorithm, Gr\"obner basis, re-encoding, ideal multiplication.
%\end{IEEEkeywords}

%\IEEEpeerreviewmaketitle
\sloppy

\section{Introduction}
The Guruswami-Sudan list decoding algorithm \cite{guruswami98improved} is one of the most powerful decoding methods for Reed-Solomon codes.
Its complexity is known to be polynomial. However, the degree of the polynomial turns out to be too high. Therefore, computationally
efficient algorithms are needed in order to obtain a practical implementation of this method. 

The most computationally intensive step of the Guruswami-Sudan algorithm is construction of a bivariate polynomial
 passing through a number of points with a given multiplicity. In this paper  a novel reduced complexity interpolation algorithm
is presented. It is based on the well-known binary exponentiation
method, so we call it binary interpolation algorithm. The algorithm exploits the relationship between the Gr\"obner bases
%\footnote{Gr\"obner bases can be considered as an instance of Shirshov bases \cite{shirshov65algorithmicEng}, which were introduced originally for the case of Lie algebras.} 
of zero-dimensional
ideals and appropriate modules.  The key component of the proposed 
method is a novel randomized fast ideal multiplication algorithm (see Figure \ref{fFastModMul}).
We show also that the interpolation complexity can be further reduced  by
integrating the proposed method  with the re-encoding approach \cite{koetter2003complexity,koetter2003efficient}. 

The paper is organized as follows. Section \ref{sPreliminaries} presents a simple derivation of the Guruswami-Sudan algorithm 
and all necessary background. Section \ref{sRootMulInc} introduces the novel interpolation algorithm. Numeric performance
results are given in  Section \ref{sNumRes}.
 Finally, some conclusions are drawn.
\section{Notation}
\begin{itemize}
%\item $\F[x_1,\ldots,x_n]$ is the ring of polynomials in variables $x_1,\ldots,x_n$ with coefficients in field $\F$.
\item $\left<Q_i(x,y),0\leq i\leq v\right>=\displaystyle \{\sum_{i=0}^v p_i(x,y)Q_i(x,y)|p_i(x,y)\in \F[x,y]\}$ is the ideal generated by $Q_i(x,y)$.
\item $[Q_i(x,y),0 \leq i \leq v]=\displaystyle\{\sum_{i=0}^v p_i(x)Q_i(x,y)|p_i(x)\in \F[x]\}$ is the module generated by $Q_i(x,y)$.
\item $Q^{[j_1,j_2]}(x_i,y_i)=\displaystyle\sum_{j_1'\geq j_1}\sum_{j_2'\geq 
j_2}{j_1'\choose j_1}{j_2'\choose 
j_2}q_{j_1'j_2'}x_i^{j_1'-j_1}y_i^{j_2'-j_2}$ is the Hasse derivative of 
$Q(x,y)$ at point $(x_i,y_i)$. 
\item $Q(x_i,y_i)=0^r$ means that  $Q(x,y)$ has a root of multiplicity {\em at least} $r$ in $(x_i,y_i)$, i.e. $Q^{[j_1,j_2]}(x_i,y_i)=0, j_1+j_2<r$.
\item $I_r=\{Q(x,y)\in \F[x,y]|Q(x_i,y_i)=0^r,1 \leq i \leq n\}$ is the ideal of polynomials having roots of multiplicity at least $r$ at  points $(x_i,y_i), 1 \leq i \leq n$.
%\item $\Wdeg_{(a,b)} Q(x,y)$ is $(a,b)$-weighted degree of polynomial $Q(x,y)$.
\item $M_{r,\rho}=\{Q(x,y)\in I_r|\Wdeg_{(0,1)}Q(x,y)<\rho\}$.% is the module of polynomials having roots of multiplicity $r$ and $y$-degree less than $\rho$.
\item $\lt Q(x,y)$  is the leading term of $Q(x,y)$ with respect to some term ordering.
%\item $\ydeg Q(x,y)=j$ iff $\lt Q(x,y)=ax^uy^j$ for some $a\in \F$ and $u\in \Z$.
\item $|\mathcal B|$ is the dimension of vector $\mathcal B$.
\item $\Delta(\mathcal B)=\sum_{j=0}^s t_j$, where $\mathcal 
B=(B_0(x,y),\ldots,B_s(x,y))$ is a Gr\"obner basis of some module, and $\lt B_j(x,y)=a_jx^{t_j}y^j$.
\end{itemize}

\section{Preliminaries}
\label{sPreliminaries}
This section introduces some background information on the Guruswami-Sudan list decoding method,
associated computational algorithms, and various algebraic concepts used in this paper.
\subsection{Term orderings}
Multivariate polynomials are extensively used in this paper, so one needs 
to introduce monomial orderings to deal with them.
$(a,b)$-weighted degree of a monomial $cx^iy^j$  equals $ai+bj$.  
$(a,b)$-weighted degree $\Wdeg_{(a,b)} Q(x,y)$ of a polynomial $Q(x,y)$ equals 
to the maximum of $(a,b)$-weighted degrees of its non-zero terms.  
Weighted degree can be used to define a term ordering. $(a,b)$-weighted degree
 lexicographic ordering is defined as
$cx^iy^j\prec dx^py^q \Leftrightarrow (ai+bj<ap+bq)\vee(ai+bj=ap+bq)\wedge( cx^iy^j\prec_{lex} 
dx^py^q) $. Lexicographic ordering is defined as $cx^iy^j\prec_{lex} dx^py^q \Leftrightarrow (j<q)\vee (j=q)\wedge (i<p)$. 
Leading term $\lt Q(x,y)$ of a polynomial $Q(x,y)=\sum q_{ij}x^iy^j$ is given by $\displaystyle\arg \max_{q_{ij}\neq 0} q_{ij}x^iy^j$.
Multivariate polynomials can be ordered according to their leading terms. 
\subsection{Guruswami-Sudan algorithm}
The Guruswami-Sudan algorithm addresses the problem of list decoding of 
$(n,k,n-k+1)$ Reed-Solomon code over field $\F$. That is, given a
received vector $(y_1,\ldots,y_n)$, it finds all message polynomials 
$f(x)$, such that $\deg f(x)<k$ and $f(x_i)=y_i$ for at least $\tau$ 
distinct code locators $x_i$ \cite{guruswami98improved}. This is 
accomplished by  constructing a polynomial $Q(x,y)$, such that 
$Q(x_i,y_i)=0^r, \wdeg Q(x,y)\leq l$, $\Wdeg_{(0,1)}Q(x,y)<\rho$,
and factoring it. 

It is possible to show that the parameters of this algorithm must satisfy \cite{nielsen98decoding}
\begin{eqnarray}
\label{mRhoBound}
\frac{\rho(\rho-1)}{2}&\leq& \frac{nr(r+1)}{2(k-1)}<\frac{\rho(\rho+1)}{2},\\
\label{mWDegBound}
l&=&\left\lfloor\frac{nr(r+1)}{2\rho}+\frac{(\rho-1)(k-1)}{2}\right\rfloor,\\
\label{mMinNumOfNonErrors}
\tau&=&\left\lfloor\frac{l}{r}\right\rfloor+1>\sqrt{n(k-1)}.
\end{eqnarray}

%The smallest number of non-erroneous positions such that list decoding with the above algorithm
% is still possible is $\tau=\lceil \sqrt{n(k-1)}\rceil$ \cite{guruswami98improved}. It is possible to solve \eqref{mRhoBound} for $\rho$;
%let $\rho(n,k,r)$ be the solution of this inequality for given values of $n,r,k$.

%In this paper, only the interpolation step is considered.

\subsection{Interpolation}
\label{sInterpolationReview}
Construction of a polynomial $Q(x,y)$ 
 turns out to be the most computationally expensive step of the Guruswami-Sudan algorithm.  
This section presents an overview of two existing algorithms for the interpolation problem. The first one will be used to derive some important
properties of the underlying algebraic structures, and the second will be used as a component of the proposed method.

Observe that the set of polynomials $I_r=\{Q(x,y)\in \F[x,y]|Q(x_i,y_i)=0^r,1 \leq i \leq n\}$ is an ideal. The smallest non-zero polynomial
of this ideal with respect to $(1,k-1)$-weighted degree lexicographic ordering must satisfy the constraints of the Guruswami-Sudan 
algorithm. Such a polynomial is guaranteed to appear in the Gr\"obner basis of $I_r$
with respect to this term ordering \cite{sauer1998polynomial}. However, it 
turns out to be easier to construct a Gr\"obner basis of the module\footnote{The concept of module is similar to the concept of linear vector space, except 
that the former one is based on a ring, while the latter is based on a field.} $M_{r,\rho}=\{Q(x,y)\in I_r|\Wdeg_{(0,1)}Q(x,y)<\rho\}$. 
\subsubsection{Iterative interpolation algorithm}
\label{sIIA}
The algorithm shown in Figure \ref{fIIA} constructs $\rho$  non-zero polynomials $Q_j(x,y), 0 \leq j \leq \rho-1,$ such that 
$Q(x_i,y_i)=0^r, \lt Q_j(x,y)=a_jx^{t_j}y^j$,  and $t_j$ are the smallest 
possible integers \cite{nielsen98decoding,koetter1996fast,okeefe2002grobner}. 
\begin{figure}
\begin{algorithm}{IterativeInterpolation}{n,\{(x_i,y_i),1 \leq i \leq n\},r,\rho}
\begin{FOR}{i\=0  \TO \rho-1}
Q_i(x,y)\=y^i;
\end{FOR}\\
\begin{FOR}{i\=1  \TO n}
\begin{FOR}{\beta\=0 \TO r-1}
\begin{FOR}{\alpha\=0 \TO r-\beta-1}
\Delta_j\=Q_j^{[\alpha,\beta]}(x_i,y_i), 
0 \leq j \leq \rho-1\\
\displaystyle m\=\arg\min_{j:\Delta_j\neq 0}Q_j(x,y)\\
\begin{FOR}{j\neq m}
Q_j(x,y)\=Q_j(x,y)-\frac{\Delta_j}{\Delta_{m}}Q_{j_0}(x,y)
\end{FOR}\\
Q_{m}(x,y)\=Q_{m}(x,y)(x-x_i);
\end{FOR}
\end{FOR}
\end{FOR}\\
\RETURN{(Q_0(x,y),\ldots,Q_{\rho-1}(x,y))}
\end{algorithm}
\caption{Iterative interpolation algorithm (IIA)}
\label{fIIA}
\end{figure}
%\begin{lemma}
%\label{lIIABasis}
%Let $Q_j(x,y), 0 \leq j \leq \rho-1$ be the polynomials constructed by IIA for a given set of interpolation points $(x_i,y_i), 1 \leq i \leq n$, parameters $r$ and $\rho$.
%Then any polynomial $Q(x,y)$, such that $Q(x_i,y_i)=0^r, 1 \leq i \leq n$, and $\Wdeg_{(0,1)}Q(x,y)<\rho$ can be represented 
%as $$Q(x,y)=\sum_{j=0}^{\rho-1} p_j(x)Q_j(x,y),$$
%and $\lt Q(x,y)=\lt p_j(x)\lt Q_j(x,y)$ for some $j$.
%\end{lemma}
%\begin{proof}
%For a proof see \cite{ma2004divide,trifonov2007interpolationEng}.
%\end{proof}
These polynomials represent a Gr\"obner basis of the
module $M_{r,\rho}=\{Q(x,y)\in I_r|\Wdeg_{(0,1)}Q(x,y)<\rho\}$ \cite{ma2004divide,trifonov2007interpolationEng}. 
In the context of list decoding one has to use $(1,k-1)$-weighted lexicographic ordering. The solution of the interpolation
problem is given by the smallest polynomial in the obtained vector $(Q_0(x,y),\ldots,Q_{\rho-1}(x,y))$.
It can be seen that the complexity of IIA   is given by  $O(n^2r^4\rho)$.

It will be sometimes convenient to represent a vector of polynomials $A_i(x,y)=\sum_{j=0}^{t}y^ja_{ji}(x), 0 \leq i \leq s,$ 
 as $(1,y,\ldots,y^{t})\underbrace{\begin{pmatrix}a_{00}(x)&a_{01}(x)&\ldots&a_{0s}(x)\\
a_{10}(x)&a_{11}(x)&\ldots&a_{1s}(x)\\
\vdots&\vdots&\ddots&\vdots\\
a_{t0}(x)&a_{t1}(x)&\ldots&a_{ts}(x)
\end{pmatrix}}_{\mathcal A(x)}$, where $\mathcal A(x)$ is a $(t+1)\times (s+1)$ polynomial matrix.
\begin{lemma}
\label{lDegDet}
Let $\mathcal Q=(Q_0(x,y),\ldots,Q_{\rho-1}(x,y))$ be a vector of polynomials constructed by IIA for the input $(n,\{(x_i,y_i), 1 \leq i \leq n\},r,\rho)$.
 Then $$\deg \det \mathcal Q(x)=\Delta (\mathcal Q)=n\frac{r(r+1)}{2},$$
where  $\mathcal Q(x)$ is the corresponding polynomial matrix.
\end{lemma}
\begin{proof}
Observe that at each iteration
 of IIA the $x$-degree of exactly one polynomial is increased by one. Hence, 
the sum of leading term $x$-degrees of all polynomials after algorithm termination is equal to the number of partial Hasse derivatives 
forced to be zero. On the other hand, this algorithm can be interpreted as construction of the polynomial matrix 
\begin{equation}
\label{mDeltaProd}
\mathcal Q(x)=\prod_{i=1}^n\prod_{\alpha+\beta<r}\delta^{(i,\alpha,\beta)},
\end{equation}
where $$
%\label{mStepMatrix}
\delta^{(i,\alpha,\beta)}=\begin{pmatrix}
1&0&\ldots&0&\ldots&0\\
0&1&\ldots&0&\ldots&0\\
%\ldots&\ldots&\ldots&\ldots&\ldots&\ldots\\
\hdotsfor{6}\\
-\frac{\Delta_0}{\Delta_{m}}&-\frac{\Delta_1}{\Delta_{m}}&\ldots&x-x_i&\ldots&-\frac{\Delta_{\rho-1}}{\Delta_{m}}\\
%\ldots&\ldots&\ldots&\ldots&\ldots&\ldots\\
\hdotsfor{6}\\
0&0&\ldots&0&\ldots&1
\end{pmatrix},$$ 
and $m=m(i,\alpha,\beta)$ is the index of the smallest polynomial  selected on line 7 of the algorithm. Obviously, $\det \delta^{(i,\alpha,\beta)}=\gamma_{i,\alpha,\beta}(x-x_i)$ for
some non-zero $\gamma_{i,\alpha,\beta}$, and the number of terms in \eqref{mDeltaProd} is again equal to the number of Hasse derivatives forced to be zero. 
\end{proof}

%\begin{lemma}
%\label{lIIABasis}
%Let $Q_j(x,y), 0 \leq j \leq \rho-1$ be the polynomials constructed by IIA for a given set of interpolation points $(x_i,y_i), 1 \leq i \leq n$, parameters $r$ and $\rho$.
%Then any polynomial $Q(x,y)$, such that $Q(x_i,y_i)=0^r, 1 \leq i \leq n$, and $\Wdeg_{(0,1)}Q(x,y)<\rho$ can be represented 
%as $$Q(x,y)=\sum_{j=0}^{\rho-1} p_j(x)Q_j(x,y),$$
%and $\lt Q(x,y)=\lt p_j(x)\lt Q_j(x,y)$ for some $j$.
%\end{lemma}
%\begin{proof}
%For a proof see \cite{ma2004divide,trifonov2007interpolationEng}.
%\end{proof}

Observe that for a fixed term ordering
there may exist many different Gr\"obner bases of a module.  However, they share 
the following common property. 
\begin{lemma}
\label{lLTSum1}
Let $\mathcal B=(B_0(x,y),\ldots,B_{\rho-1}(x,y))$ be a Gr\"obner basis of the  module $M_{r,\rho}$. 
%Let $\lt B_j(x,y)=a_jx^{t_j}y^j, 0 \leq j \leq \rho-1$ be the leading terms of polynomials $B_j(x,y)$. 
Then $\Delta(\mathcal B)=\frac{nr(r+1)}{2}$.
\end{lemma}
\begin{proof}
Let $Q_j(x,y),0 \leq j \leq \rho-1$ be the Gr\"obner basis of $M_{r,\rho}$ constructed by IIA for the same term ordering. 
Then  $\lt Q_j(x,y)|\lt B_j(x,y)$ and $\lt B_j(x,y)|\lt Q_j(x,y)$. 
This means that the leading terms of $B_j(x,y)$ and $Q_j(x,y)$ are the same up to a constant in $\F$, and the statement
follows from Lemma \ref{lDegDet}.
\end{proof}

%The Gr\"obner basis $\mathcal B$ of $M_{r,\rho}$ constructed by IIA for 
%the case of $(1,k-1)$-weighted degree lexicographic 
%ordering is not, in general, a Gr\"obner basis of $I_r$. However, IIA can be used to obtain 
%a Gr\"obner basis $\mathcal B'$ of $M_{r,\rho+1}$, where 
%$\rho=\rho(n,k,r)$. This basis also satisfies $\Delta (\mathcal 
%B')=\frac{nr(r+1)}{2}$. Furthermore, $\lt Q_\rho(x,y)=y^\rho$. 
%Indeed, if this is not true, $Q_{\rho}(x,y)$ has been selected at some 
%iteration as the smallest polynomial at step 7 of the algorithm. This 
%implies that $\wdeg Q_j(x,y)=t_j+j(k-1)\geq \rho(k-1), 0 \leq j \leq \rho,$ at that 
%iteration. Hence $\Delta(\mathcal B')\geq \sum_{j=0}^\rho t_j\geq 
%\sum_{j=0}^\rho 
%(\rho-j)(k-1)=\frac{\rho(\rho+1)}{2}(k-1)>n\frac{r(r+1)}{2}$. It follows 
%also that other polynomials in $\mathcal B'$ are the same as in $\mathcal 
%B$.  
%Therefore, leading term of any polynomial in $I_r$ is divisible by leading 
%terms of polynomials in $\mathcal B'$, i.e. $\mathcal B'$ is a Gr\"obner 
%basis of $I_r$.  This will be discussed  in more details in the proof of 
%Lemma \ref{lLTSum} below. 

\subsubsection{Transformation of module basis}
\label{sBasisChange}
It was shown in \cite{lee2006interpolation,alekhnovich2005linear,lee2008list,trifonov2008relationship} that the ideal of interpolation
polynomials  $I_r$  is generated by 
\begin{equation}
\label{mMultipleRootIdeal}
\Pi_{r,j}(x,y)=(y-T(x))^j\phi^{r-j}(x),0 \leq j \leq r,
\end{equation}
 where $T(x_i)=y_i, 1 \leq i \leq n,$ 
and $\phi(x)=\prod_{i=1}^n(x-x_i)$. Hence, the basis of the module $M_{r,\rho}$ is given by $\mathcal L=(\Pi_{r,0}(x,y),\ldots,\Pi_{r,r}(x,y),\Pi_{r,r+1}(x,y),\ldots,\Pi_{r,\rho-1}(x,y))$, where
\begin{equation}
\label{mMultipleRootIdeal1}
\Pi_{r,r+j}(x,y)=y^j\Pi_{r,r}(x,y),0< j<\rho-r.
\end{equation}
%\begin{equation}
%\label{mLexBasis}
%\Pi_j(x,y)=\begin{cases}
%(y-T(x))^j\phi^{r-j}(x),&0\leq j\leq r\\
%y^{j-r}(y-T(x))^r,&r\leq j<\rho
%\end{cases}.
%\end{equation}
\begin{lemma}
\label{lBuchbergerSPair}
The polynomials $(S_0(x,y),\ldots,S_{s-1}(x,y))$ represent a Gr\"obner basis of the module $M=\{\sum_{j=0}^{s-1} S_j(x,y)a_j(x)|a_j(x)\in \F[x]\}$ if 
$\ydeg S_i(x,y), 0 \leq i \leq s-1$ are distinct values. 
%Here $\ydeg Q(x,y)=j$ iff $\lt Q(x,y)=ax^uy^j$ for some $a\in \F$ and $u\in \N$.
\end{lemma}
\begin{proof}
The lemma follows from the Buchberger S-pair criterion \cite{becker93grobner}.
\end{proof}
The above described basis $\mathcal L$ has to be transformed 
into a Gr\"obner one with respect to  $(1,k-1)$-weighted 
degree lexicographic monomial ordering. This can be 
done with the  algorithm given in  
\cite{alekhnovich2005linear,lee2008list}, which can be considered as a simplified instance of the Buchberger
algorithm. It is convenient to present it here in a slightly modified form.
Namely,  this algorithm takes as input some polynomial $P(x,y)$, Gr\"obner basis $(S_0(x,y),\ldots,S_{i-1}(x,y))$ of some module
 $M\subset\F[x,y]$, and constructs a  Gr\"obner basis of module $M'=\{Q(x,y)+a(x)P(x,y)|Q(x,y)\in M, a(x)\in \F[x]\}$. The algorithm is shown
in Figure \ref{fGenEuclid}.
\begin{figure}
\begin{algorithm}{Reduce}{(S_0(x,y),\ldots,S_{i-1}(x,y)),P(x,y)}
S_i(x,y)\=P(x,y)\\
\begin{WHILE}{\exists j:(0\leq j<i) \wedge (\ydeg S_j(x,y)=\ydeg S_i(x,y))}
\begin{IF}{\lt S_i(x,y)|\lt S_j(x,y)}
W(x,y)\=S_j(x,y)-\frac{\lt S_j(x,y)}{\lt S_i(x,y)}S_i(x,y)\\
S_j(x,y)\=S_i(x,y)\\
S_i(x,y)\=W(x,y)
\ELSE
S_i(x,y)\=S_i(x,y)-\frac{\lt S_i(x,y)}{\lt S_j(x,y)}S_j(x,y)
\end{IF}
\end{WHILE}\\
\begin{IF}{S_i(x,y)=0}
i\=i-1
\end{IF}\\
\RETURN (S_0(x,y),\ldots,S_i(x,y))
\end{algorithm}
\caption{Multi-dimensional Euclidean algorithm}
\label{fGenEuclid}
\end{figure}
\begin{lemma}
\label{lReduceGroebner}
Let $S_j(x,y),0 \leq j \leq i-1$ be the polynomials such that $\lt S_j(x)=\alpha_j x^{t_j}y^j, \Wdeg_{(0,1)}S_j(x,y)<i$. 
Then the $Reduce$ algorithm  constructs a Gr\"obner basis of the module $M=[S_0(x,y),\ldots,S_{i-1}(x,y),P(x,y)]$.
\end{lemma}
\begin{proof}
%Observe that the algorithm applies at each iteration  invertible transformations to the polynomials being processed, so 
%the obtained set of polynomials is indeed a basis. 
%Furthermore, at each iteration leading term of one polynomial is cancelled. Hence, the algorithm terminates. Leading
%terms of the obtained polynomials have different $y$-degrees, so they represent a Gr\"obner basis of $M$ by Lemma \ref{lBuchbergerSPair}.
This statement follows from  Lemma \ref{lBuchbergerSPair} and 
invertibility of transformations used by the algorithm.
\end{proof}
The required Gr\"obner basis is obtained as $\mathcal S_{\rho-1}$, where 
\begin{equation}
\label{mGBExtension}
\mathcal S_j=Reduce(\mathcal S_{j-1},\Pi_{r,j}(x,y)), \mathcal S_0=(\Pi_{r,0}(x,y)).
\end{equation}
 The complexity of this method is given by $O(n^4k^{-2}r^5)$  \cite{lee2008list}.
Curiously, if $(1,k-1)$-weighted degree lexicographic ordering is used and  
$r=1,\rho=2$, it reduces to  the Gao decoding
 method \cite{gao2003new,fedorenko2005simple}, with function $Reduce$ being the standard
extended Euclidean algorithm with early termination
condition. Therefore, $Reduce$  will be referred to as the multi-dimensional Euclidean algorithm.

\section{Binary interpolation algorithm}
\label{sRootMulInc}
This section introduces a novel interpolation algorithm. The main idea of this algorithm is to construct a sequence
of ideals and  modules of polynomials having roots $(x_i,y_i)$ with increasing  multiplicity. The proposed method
can be considered as an application of the  well-known binary exponentiation algorithm to zero-dimensional ideals.
\subsection{Interpolation via ideal multiplication}
The main drawback of the method given by \eqref{mGBExtension} is that one has to manipulate with the polynomials having large
common divisors. For example, $\mathcal S_1=Reduce((\Pi_{r,0}(x,y)),\Pi_{r,1}(x,y))=Reduce((\phi^r(x)),\phi^{r-1}(x)(y-T(x)))=\phi^{r-1}(x)Reduce((\phi(x)),y-T(x))$.
Furthermore, polynomial exponentiation is used in 
\eqref{mMultipleRootIdeal}. The method proposed in this paper avoids both reducing the polynomials 
with large GCD, and computing large powers of polynomials. This is 
achieved by first constructing Gr\"obner bases for small root multiplicities, and using them to 
obtain bases for larger root multiplicities.

\begin{lemma}
Let $I_r=\{Q(x,y)\in \F[x,y]|Q(x_i,y_i)=0^r, 1 \leq i \leq n\}$. Then $I_{r_1+r_2}=I_{r_1}\cdot I_{r_2}$.
\end{lemma}
\begin{proof}
%For any $Q(x,y)\in I_r$ one has $Q(x,y)=\sum_{j_1+j_2\geq r}q_{j_1j_2}(x-x_i)^{j_1}(y-y_i)^{j_2}$. Hence, $I_{r_1}\cdot I_{r_2}=\{\sum_s Q_s(x,y)P_s(x,y)|Q_s(x,y)\in I_{r_1},P_s(x,y)\in I_{r_2}\}=\{F(x,y)=\sum_{j_1+j_2\geq r_1+r_2}f_{j_1j_2}(x-x_i)^{j_1}(y-y_i)^{j_2}\}\subset I_{r_1+r_2}$. Furthermore, $I_{r_1+r_2}=\left<(y-T(x))^j\phi^{r-j}(x),0 \leq j \leq r_1+r_2\right>$, and it is always possible to find $j_1,j_2: 0\leq j_1\leq r_1, 0\leq j_2\leq r_2: j_1+j_2\leq r_1+r_2$. Hence, the generating
%elements of $I_{r_1+r_2}$ can be represented as a product of generating elements of $I_{r_1}$ and $I_{r_2}$, i.e. $I_{r_1+r_2}\subset I_{r_1}\cdot I_{r_2}$.
$I_{r_1+r_2}=<(y-T(x))^j\phi^{r_1+r_2-j}(x),0 \leq j \leq r_1+r_2>=<(y-T(x))^{j_1+j_2}\phi^{r_1-j_1+r_2-j_2}(x),j_1=0,\ldots,r_1,j_2=0,\ldots,r_2>=<(y-T(x))^{j_1}\phi^{r_1-j_1}(x),j_1=0,\ldots,r_1>\cdot 
<(y-T(x))^{j_2}\phi^{r_2-j_2}(x),j_2=0,\ldots,r_2>=I_{r_1}\cdot I_{r_2}$.
\end{proof}
This lemma implies that $I_r=\underbrace{I_1\cdot I_1\cdots I_1}_{r\text{times}}=I_1^r$.
One can avoid repeated calculations and reduce the overall number of calls to the $Reduce$ algorithm
by using the  binary exponentiation method   \cite{KnuthArt2}. Namely, one can compute 
%$$I_r=I_1^{r_0}I_2^{r_1}I_4^{r_2}\cdots I_{2^d}^{r_d},$$ 
$$I_r=I_1^r=(\ldots((I_1^2\cdot I_1^{r_{m-1}})^2\cdot I_1^{r_{m-2}})^2\cdot I_1^{r_{m-3}}\cdots I_1^{r_{1}})^2\cdot I_1^{r_{0}}$$
where   $r=\sum_{j=0}^m r_j2^j, r_m=1$, $I^2=I\cdot I$, $I^0=\F[x,y]$, and $I\cdot \F[x,y]=I$. 

The key problem addressed in this paper is 
how to construct efficiently a Gr\"obner basis of the product of ideals 
$I'=\left<P_0(x,y),\ldots,P_u(x,y)\right>$ and $I''={\left<S_0(x,y),\ldots,S_v(x,y)\right>}$.  
The standard way is given by 
\begin{equation}
\label{mIdealProduct}
I'\cdot I''=\left<P_i(x,y)S_j(x,y),0 \leq i \leq u,0 \leq j \leq v\right>,
\end{equation}
 i.e. to compute pairwise products of all basis elements of the 
ideals being multiplied. This requires $(u+1)(v+1)$ bivariate polynomial
multiplications, and the basis of $I'\cdot I''$ obtained in such way is extremely redundant.
Furthermore, Buchberger algorithm must be used in order to obtain a 
Gr\"obner basis of $I_r$.

To the best of author knowledge, the problem of efficient ideal multiplication was not considered in the literature, except in \cite{trifonov2007interpolationEng}, 
where multiplication of zero-dimensional co-prime ideals was reduced  to linear convolution. However, the ideals considered in this paper are not co-prime. 

This problem can be again solved by constructing at each step of the binary exponentiation algorithm a basis of the module of polynomials
with limited $(0,1)$-weighted degree. 
\begin{lemma}
\label{lLTSum}
Consider the polynomials $P_j(x,y): P_j(x_i,y_i)=0^s, 1 \leq i \leq n, 0 \leq j \leq m$, such that $\lt P_j(x,y)=a_jx^{t_j}y^j$, $\Wdeg_{(0,1)}P_j(x,y)\leq m$,
 $t_m=0$, 
and 
\begin{equation}
\label{mLTSum}
\Delta\left((P_0(x,y),\ldots,P_m(x,y))\right)=\frac{ns(s+1)}{2}.
\end{equation}
 Then $I_s=\left<P_j(x,y),0 \leq j \leq m\right>$, and the polynomials $P_j(x,y)$ constitute a Gr\"obner basis of this ideal.
\end{lemma}
\begin{proof}
Observe that the polynomials $P_j(x,y)$ represent a Gr\"obner basis of some module by lemma \ref{lBuchbergerSPair}.
Obviously, $\left<P_j(x,y),0 \leq j \leq m\right>\subset I_s$. Suppose  that the polynomials $P_j(x,y)$ do not constitute a Gr\"obner basis
of $I_s$. That is, there exists $S(x,y)\in I_s: S(x,y)=\sum_{j=0}^m q_j(x,y)P_j(x,y)+R(x,y)$, where the terms 
of $R(x,y)$ are not divisible by $\lt P_j(x,y)$, i.e. $\Wdeg_{(0,1)}R(x,y)<m$ and $\lt R(x,y)=\beta x^uy^v, u<t_v$. Observe
that $R(x,y)\in M_{s,m}$. This means that the polynomials $P_j(x,y), 0 \leq j \leq m-1,$ do not represent a Gr\"obner basis of module $M_{s,m}$.  The true 
Gr\"obner basis of this module should consist of smaller polynomials, i.e. the sum of $x$-degrees of their leading terms should be less
than $\frac{ns(s+1)}{2}$. But this contradicts to Lemma \ref{lLTSum1}. Hence, $R(x,y)=0$ and $P_j(x,y)$
constitute a Gr\"obner basis of $I_s$.  
\end{proof}
Observe that there may exist Gr\"obner bases of $I_s$ not satisfying the constraints of this lemma.

Let $I_{r_1}=\left<P_j(x,y), 0 \leq j \leq u\right>, I_{r_2}=\left<S_i(x,y),0 \leq i \leq v\right>$ be the ideals given by their Gr\"obner bases satisfying the above lemma.
One can   construct a Gr\"obner basis of the   product $I_{r_1+r_2}$ of these ideals as follows. Let $(m_j',m_j''), 0 \leq j \leq (u+1)(v+1)-1$ be a
 sequence of distinct pairs of
integers such that $0\leq m_j'\leq u, 0\leq m_j''\leq v$, and $\lt \left(P_{m_j'}(x,y)S_{m_j''}(x,y)\right)=\alpha_j x^{t_j}y^j$ for $j\leq u+v$.
Let 
\begin{equation}
\label{mInitialBasis}
\mathcal B_{u+v}=(P_{m_j'}(x,y)S_{m_j''}(x,y), 0 \leq j \leq u+v)
\end{equation}
 be a basis of some submodule of $M_{r_1+r_2,u+v+1}$. By Lemma \ref{lBuchbergerSPair} it is a Gr\"obner basis
of this submodule. 
It can be seen that  $\Delta(\mathcal B_{u+v})=\sum_{j=0}^{u+v} t_j \geq \frac{n(r_1+r_2)(r_1+r_2+1)}{2}$. 
%Indeed,
%let $\mathcal B=(Q_0(x,y),\ldots,Q_{u+v}(x,y))$ be a Gr\"obner basis of $M_{r_1+r_2,u+v+1}$. By Lemma \ref{lLTSum1},
%$\Delta(\mathcal B)=\frac{n(r_1+r_2)(r_1+r_2+1)}{2}$. If $\Delta(\mathcal B_{u+v})<\Delta(\mathcal B)$, $\mathcal B_{u+v}$ contains at
% least one polynomial $\hat Q_j(x,y)=P_{m_j'}(x,y)S_{m_j''}(x,y): \lt \hat Q_j(x,y)=\alpha_j x^{t_j}y^j, t_j<s_j$, where
%$\lt Q_j(x,y)=x^{s_j}y^j$. But this is impossible, since $\hat Q(x,y)\in M_{r_1+r_2,u+v+1}$ and $\mathcal B
%$ is the Gr\"obner basis of this module. 

Let 
\begin{equation}
\label{mBasisInduction}
\mathcal B_{j}=Reduce(\mathcal B_{j-1},P_{m_j'}S_{m_j''}), j> u+v.
\end{equation}
The  $Reduce$ algorithm attempts to cancel the leading terms of the provided polynomials, so $\Delta(\mathcal B_{j+1})\leq \Delta(\mathcal B_j)$.
As soon as one obtains $\Delta(\mathcal B_j)=\frac{n(r_1+r_2)(r_1+r_2+1)}{2}$,  $\mathcal B_j$ is a Gr\"obner basis of $I_{r_1+r_2}$.
\begin{lemma}
\label{lModuleProduct}
$M_{r_1+r_2,u+v+1}$ is generated by $\mathcal B_{(u+1)(v+1)-1}$.
\end{lemma}
\begin{proof}
Consider $Q(x,y)\in I_{r_1+r_2}$, such that $\Wdeg_{(0,1)} Q(x,y)\leq 
u+v$. Any such polynomial can be represented as $Q(x,y)=\displaystyle\sum_{j=0}^uP_j(x,y)\sum_{i=0}^v 
q_{ji}(x,y)S_i(x,y)$. Inner sum is an element of $I_{r_2}$. Since the polynomials $S_i(x,y), 0 \leq i \leq v,$ are
a Gr\"obner basis of  $M_{r_2,v+1}$ and $I_{r_2}$, one can
 use the multivariate polynomial division algorithm to obtain $Q(x,y)=\displaystyle \sum_{j=0}^uP_j(x,y)\left(\sum_{i=0}^{v-1} w_{ji}(x)S_i(x,y)+S_v(x,y)\sum_{i\geq v}y^{i-v}\tilde w_{ji}(x)\right)$. Similarly, $P(x,y)=\displaystyle \sum_{j=0}^u P_j(x,y) \sum_{i\geq v}y^{i-v}\tilde w_{ji}(x)$ is 
in $I_{r_1}$, and the multivariate
division algorithm leads to $P(x,y)=\displaystyle \sum_{j=0}^u w_{jv}(x)P_j(x,y)+P_u(x,y)\sum_{j>u} y^{j-u}w_{jv}(x)$. Hence, 
$Q(x,y)=\displaystyle \sum_{j=0}^u\sum_{i=0}^{v} P_j(x,y)S_i(x,y)w_{ji}(x)+P_u(x,y)S_v(x,y)y\overline w(x,y)$.
Last term in this expression is zero, since $Q(x,y)$ does not contain 
any monomials $ax^py^q$ with $q>u+v$, so 
$M_{r_1+r_2,u+v+1}=[P_j(x,y)S_i(x,y), 0 \leq j \leq u,0 \leq i \leq v]$.
%Observe that for each $j$ the output of the  $Reduce$ algorithm is a Gr\"obner basis of some module, i.e. the set of smallest
%polynomials which can be obtained as linear combinations over $\F[x]$ of the provided polynomials. Hence, it remains to show that 
%any polynomial in $M_{r_1+r_2,u+v+1}$ can be represented as $Q(x,y)=\sum_{i=0}^u\sum_{j=0}^v q_{ij}(x)P_i(x,y)S_j(x,y)$. Indeed, the polynomials 
%$P_i(x,y)$ and $S_j(x,y)$ are the bases of $M_{r_1,u+1}$ and $M_{r_2,v+1}$, respectively. Hence, 
%there exist polynomials 
%$s_{ji}(x)$ and $p_{ji}(x)$, such that $$\Pi_{r_1,j}(x,y)=\sum_{i=0}^u p_{ji}(x)P_i(x,y), 0 \leq j \leq u$$ and $$\Pi_{r_2,j}(x,y)=\sum_{i=0}^v s_{ji}(x)S_i(x,y), 0 \leq i \leq v.$$
%One can obtain the basis polynomials for $M_{r_1+r_2,u+v+1}$ as  $\Pi_{r_1+r_2,i}(x,y)=\Pi_{r_1,i-j}(x,y)\Pi_{r_2,j}(x,y)$, where $\max(0,i-r_1)\leq j\leq r_2$
%for $0\leq i\leq r_1+r_2$ and $r_2\leq j\leq \min(v,i-r_1)$ for $r_1+r_2\leq i\leq u+v$.  Hence, $\Pi_{r_1+r_2,i}$ can be obtained 
%as a linear combination over $\F[x]$ of polynomials $P_{i'}(x,y)S_{i''}(x,y)$, which belong to the module generated by $\mathcal B_{(u+1)(v+1)-1}$.
\end{proof}
The lemma states that for any suitable polynomial $Q(x,y)\in I_{r_1}I_{r_2}$ one can replace the 
bivariate polynomials $q_{ji}(x,y)$ with univariate ones $w_ {ji}(x)$. This implies that the 
sequence $\mathcal B_j$ converges eventually to the required module basis. However, the convergence turns out
to be quite slow. One may need to compute many bivariate polynomial products $P_{m_j'}S_{m_j''}$ and apply $Reduce$ algorithm to them before
the constraint \eqref{mLTSum} is satisfied.
In many cases it appears even that $\mathcal B_{j+1}=\mathcal B_j$. That is, a significant fraction  of  pairs $(m_j',m_j'')$ is useless.

Therefore we propose to replace pairwise products $P_j(x,y)S_i(x,y)$ in \eqref{mIdealProduct} with their random linear combinations 
\begin{equation}
\label{mRandomLinearComb}
Q_s(x,y)=\sum_{j=0}^u \sum_{i=0}^v \chi_{sji} P_{j}(x,y)S_{i}(x,y),
\end{equation}
 where $\chi_{sji}$ are independent random variables uniformly
distributed over $\F$. Obviously, such polynomials still generate the ideal product if the linear transformation given by $\chi_{sji}$ 
is invertible, i.e. if at least $(u+1)(v+1)$ polynomials $Q_s(x,y)$ are given. However, it turns out that in
 average one needs just a few such polynomials to obtain a basis of the ideal product. The reason is that $Q_s(x,y)$ depend on all pairwise products $P_j(x,y)S_i(x,y)$,
and a Gr\"obner basis construction algorithm (e.g. $Reduce$) can take them into account simultaneously.
This will be discussed in more details in Section \ref{sComplexity}.

However, it is impractical to construct the polynomials explicitly as given by \eqref{mRandomLinearComb}, since this requires one first to compute all pairwise 
products $P_{j}(x,y)S_{i}(x,y)$. More efficient way is to construct a sequence of bases
$${ \mathcal B_{j+1}'}=Reduce\left( \mathcal  B_j',\left(\sum_{i=0}^u\alpha_{ij}P_i(x,y)\right)\left(\sum_{i=0}^v\beta_{ij}S_i(x,y)\right)\right),$$
where $j\geq u+v$, and  $\alpha_{ij},\beta_{ij}$ are some random  values uniformly distributed over $\F$. Furthermore, we propose 
to construct the initial basis  $\mathcal  B_{u+v}'=(Q_0(x,y),\ldots,Q_{u+v}(x,y))$ as $Q_i(x,y)=P_{i-j}(x,y)S_j(x,y)$, where for each $i$
$j$ is selected so that $\lt Q_i(x,y)=a_ix^{t_i}y^i$, and the  values $t_i,0 \leq i \leq u+v$ are minimized. This reduces the number of iterations needed by the $Reduce$ algorithm.
 The proposed approach is summarized in Figure \ref{fFastModMul}. 
\begin{figure}
\begin{algorithm}{Merge}{(P_0(x,y),\ldots,P_u(x,y)),(S_0(x,y),\ldots,S_v(x,y)),\Delta_0}
\begin{FOR}{i\=0 \TO u+v}
Q_i(x,y)=\min_{0\leq j\leq v}P_{i-j}(x,y)S_j(x,y)
\end{FOR}\\
\mathcal B=(Q_0(x,y),\ldots,Q_{u+v}(x,y))\\
\begin{WHILE}{\Delta(\mathcal B)>\Delta_0}
\alpha_i\=rand(), 0 \leq i \leq u\\
\beta_j\=rand(),0 \leq j \leq v\\
Q(x,y)\=\left(\sum_{i=0}^u\alpha_iP_i(x,y)\right)\left(\sum_{i=0}^v\beta_iS_i(x,y)\right)\\
\mathcal B\=\CALL {Reduce}(\mathcal B,Q(x,y))
\end{WHILE}\\
\RETURN \mathcal B
\end{algorithm}
\caption{Construction of a Gr\"obner basis of $I=JK$ from the Gr\"obner bases of $J=\left<P_0(x,y),\ldots,P_u(x,y)\right>$ and
 $K=\left<S_0(x,y),\ldots,S_v(x,y)\right>$.}
\label{fFastModMul}
\end{figure}

\begin{thm}
\label{tMerge}
Given Gr\"obner bases $\mathcal P=(P_0(x,y),\ldots,P_u(x,y))$ and $\mathcal S=(S_0(x,y),\ldots,S_v(x,y))$ of ideals $I_{r_1}$ and $I_{r_2}$,
the result of $Merge(\mathcal P,\mathcal S,n\frac{r(r+1)}{2})$ is a Gr\"obner basis of $I_{r}$, where $r=r_1+r_2$. 
\end{thm}
\begin{proof}
%It is possible to obtain all pairwise
%products $P_{i'}(x,y)S_{i''}(x,y)$ from sufficiently many polynomials
% $$U_j(x,y)=\left(\sum_{i'=0}^u\alpha_{i'j}P_{i'}(x,y)\right)\left(\sum_{i''=0}^v\beta_{i''j}S_{i''}(x,y)\right).$$ Indeed, 
%having collected $m$ such polynomials one can construct a system of linear equations $U=AV$, where 
%$U=(U_{u+v}(x,y),\ldots,U_{u+v+m-1}(x,y))^T$, $V$ is the unknown vector of all pairwise products $P_{i'}(x,y)S_{i''}(x,y)$, 
%and $A$ is the $m\times (u+1)(v+1)$ matrix containing $\alpha_{i'j}\beta_{i''j}$ values. This system can be solved if
% the rank of $A$ is $(u+1)(v+1)$.
Observe that the sequence $B_j'$ still converges to a basis of $M_{r_1+r_2,u+v+1}$, since 
it is possible to  select  $\alpha_{ij}$ and $\beta_{ij}$ so that the linear transformation \eqref{mRandomLinearComb} given by $\chi_{sij}=\alpha_{is}\beta_{js}$ 
is invertible, provided that   sufficiently many polynomials are constructed. By 
lemma \ref{lReduceGroebner}, the $Reduce$ algorithm always produces a Gr\"obner basis of some module. By lemma \ref{lLTSum},
this basis is a Gr\"obner basis of $I_r$.
\end{proof}
\begin{rem}
$Merge$ is not guaranteed to obtain a minimal Gr\"obner basis of $I_r$. In particular, a few 
 polynomials may have $\lt B_j(x,y)=y^j$. Such polynomials are redundant, and should be eliminated, except the smallest one. 
\end{rem}

The overall interpolation algorithm is shown in Figure \ref{fInterpolationAlg}. $(1,k-1)$-weighted degree lexicographic ordering must
be used throughout this algorithm.
\begin{figure}
\begin{algorithm}{Interpolate}{((x_i,y_i),1 \leq i \leq n),r}
\phi(x)\=\prod_{i=1}^n (x-x_i)\\
T(x)\=\sum_{i=1}^n y_i\frac{\prod_{j\neq i}(x-x_j)}{\prod_{j\neq i}(x_i-x_j)}\\
\mathcal G\=(\phi(x))\\
j\=0\\
\begin{REPEAT}
\mathcal \mathcal G\=\CALL {Reduce}(\mathcal G,y^j(y-T(x)))\\
j\=j+1
\end{REPEAT}{\lt {\mathcal G}_j=  y^j}\\
\mathcal B\=\mathcal G\\
\text{Let $r=\sum_{j=0}^m r_j2^j, r_j\in \{0,1\}$}\\
R\=1\\
\begin{FOR}{j\=m-1 \TO 0 }
R\=2R\\
\mathcal B\=\CALL {Merge}(\mathcal B,\mathcal B,nR(R+1)/2)\\
\begin{IF}{r_j=1}
R\=R+1\\
\mathcal B\=\CALL {Merge}(\mathcal B,\mathcal G,nR(R+1)/2)
\end{IF}
\end{FOR}\\
\RETURN \mathcal B
\end{algorithm}
\caption{Construction of a Gr\"obner basis for $I_r$}
\label{fInterpolationAlg}
\end{figure}
Observe that in most practical cases the polynomial $T(x)$ can be constructed by using fast inverse discrete Fourier transform. FFT
can be also used in the implementation of polynomial multiplication, which is extensively used by this algorithm. 
\begin{thm}
\label{tInterpolate}
$Interpolate$  algorithm constructs a Gr\"obner basis of $I_r$ with respect to a given term
ordering.
\end{thm}
\begin{proof}
The objective of the  $REPEAT$ loop is to construct a Gr\"obner basis of $M_{1,j+1}$, such that it is also a Gr\"obner basis of $I_1$.
Any Gr\"obner basis of a zero-dimensional 
ideal must contain a polynomial $Q(x,y):\lt Q(x,y)=y^j$ for some $j$ \cite[Th. 6.54]{becker93grobner},
 so this loop terminates eventually, and  $\mathcal G$ is indeed a Gr\"obner basis of $I_1$.

Let $r'=\sum_{i=j+1}^m r_i2^{i-j-1}$. By induction, the input vectors to $Merge$  at line 14   are two copies
of a  Gr\"obner basis of $I_{r'}$. 
By Theorem \ref{tMerge} its output  is a Gr\"obner basis of $I_{2r'}$.
   Similar argument applies to line 17. Hence, at the end of each iteration of the $FOR$ loop one 
obtains a Gr\"obner basis of $I_{2r'+r_j}$. Observe also, that at the end of each iteration $R=2r'+r_j$.
\end{proof}
The interpolation polynomial needed by the Guruswami-Sudan algorithm can be found as the smallest element of the basis
produced by the $Interpolate$ algorithm.
\subsection{Complexity analysis}
\label{sComplexity}
%The complexity of the proposed interpolation algorithm depends mostly on the number of iterations needed by algorithm
%$Reduce$ to obtain a Gr\"obner basis of the module. It is quite difficult to estimate it in the general case. This
%section presents some rough approximations, which may still  be useful for comparison of different algorithms.

%Consider application of $Reduce$ to polynomials $S_0(x,y),\ldots,S_{b-1}(x,y): \lt S_i(x,y)=\alpha_i x^{t_i}y^i$, and polynomial 
%$P(x,y)=p_b(x)y^b+\ldots+p_c(x)y^c:\lt P(x,y)=\beta x^sy^u, c\leq u<b, $. Assume that $(1,k-1)$-weigted degree lexicographic ordering is used.
%In the worst-case scenario, up to $s+c(k-1)-(\deg p_b(x)+b(k-1))$ terms would be cancelled, not only in $P(x,y)$, but also in $b$ other polynomials.
%Hence, the upper bound on the number of iterations in  $Reduce$ is $b(s-(\deg p_j(x)+(j-c)(k-1)))$. The upper bound on the size of the Gr\"obner basis 
%of $I_1$ can be derived from \eqref{mRhoBound} as $|\mathcal G|\leq \left\lceil \frac{1+\sqrt{1+\frac{8n}{k-1}}}{2}\right\rceil$. Each polynomial in this
%basis has approximately $n$ terms, and the input polynomial at each iteration is given by $P_j(x,y)=y^j(y-T(x))$. Hence, at most $(j+1)(n-k)$
%iterations of $Reduce$ algorithm are needed at the $j$-th iteration of the first WHILE loop. 
% That is, the overall complexity of steps 5--7 is $O\left(n(n-k)\frac{n}{k}\right)$. Observe that this is $n/k$ times more than the complexity
%of Gao decoding algorithm, since the latter one requires constructing just the basis of $M_{1,1}$. 
Let us first estimate the convergence speed of the $Merge$ algorithm. Recall, that this algorithm constructs a Gr\"obner basis of $M_{r_1+r_2,u+v+1}$ given Gr\"obner bases
of $M_{r_1,u+1}$ and $M_{r_2,v+1}$ (in fact, $I_{r_1}$ and $I_{r_2}$). For the sake of simplicity we will estimate the probability of $M_{r_1+r_2,u+v+1}$ being
generated by $u+v+1+\delta$ polynomials given by \eqref{mRandomLinearComb}, 
 such that $\lt Q_s(x,y)=\gamma_s x^{t_s}y^s$ for $s\leq u+v$, and leading terms of summands do not cancel. The difference in the behavior of the actual algorithm $Merge$
with respect to this impractical ideal multiplication method will be discussed below.

The polynomials $Q_0(x,y),\ldots,Q_{u+v+\delta}(x,y)$ can be represented  as a $(u+v+1)\times (u+v+1+\delta)$ polynomial matrix $\mathcal Q(x)$.
If they indeed generate $M_{r_1+r_2,u+v+1}$,  then the  $(u+v+1)\times (u+v+1)$ 
polynomial matrix  $\mathcal B(x)$  corresponding  to the Gr\"obner basis $B_0(x,y),\ldots,B_{u+v}(x,y)$ of this module as constructed  by IIA satisfies 
$$\mathcal Q(x)\mathcal A(x)=\mathcal B(x),$$
 where  $\mathcal Q(x)$  is the polynomial matrix corresponding to $Q_i(x,y)$, and  $\mathcal A(x)$ is some transformation matrix. 
On the other hand, $Q_j(x,y)\in M_{r_1+r_2,u+v+1}$, i.e. $\mathcal Q(x)=\mathcal B(x)\Lambda(x)$, where
the elements of $\Lambda(x)$ matrix can be obtained by the multivariate division algorithm.
 Hence, 
\begin{equation}
\label{mBasisSystem}
\mathcal B(x)\Lambda(x)\mathcal A(x)=\mathcal B(x).
\end{equation}
 Since  $\lt B_j(x,y)=x^{t_j}y^j,0 \leq j \leq u+v$, the polynomials $B_j(x,y)$ are linearly independent over $\F[x]$,
and $\mathcal B(x)$ is invertible over the field of rational functions, so it can be cancelled in \eqref{mBasisSystem}.
Therefore, the problem reduces to estimating the probability of existence of a polynomial
matrix $\mathcal A(x)$ satisfying 
\begin{equation}
\label{mPseudoInverse}
\Lambda(x)\mathcal A(x)=I.
\end{equation}
This is a system of linear equations in terms of $\mathcal A(x)$. Observe that $\Lambda(x)$ is a full-rank matrix over $\F[x]$.
Polynomial 
solution exists if and only if the scalar matrix equations
\begin{equation}
\label{mBasisSystemEval}
 \Lambda(x_w) \mathcal A(x_w)=I
\end{equation} 
are solvable for any $x_w\in \F^*$, i.e. $\Lambda(x_w)$
 matrices have rank $u+v+1$.
%In fact, it is possible to recover $\mathcal A(x)$ from $\mathcal A(x_w)$ using standard univariate interpolation techniques. 
It is sufficient to consider only such $x_w$ that some fixed $(u+v+1)\times (u+v+1)$
submatrix $\widehat \Lambda(x)$  of $\Lambda(x)$ looses rank for $x=x_w$, i.e. the 
roots of $\det \widehat \Lambda(x)$.  Such roots are called eigenvalues of polynomial matrix $\widehat \Lambda(x)$ \cite{gohberg2009matrix}.

Let $\widehat\Lambda(x)$ be a matrix consisting of first $u+v+1$ columns
of $\Lambda(x)$. This matrix satisfies $$\mathcal B(x)\widehat \Lambda(x)=\widehat {\mathcal Q}(x),$$
where the polynomial matrix $\widehat{\mathcal Q}(x)$ corresponds to $Q_0(x,y),\ldots,Q_{u+v}(x,y)$. 
For each eigenvalue $x_w\in \F^*$  of $\widehat \Lambda(x)$ one can   identify $n_w$ linearly independent left 
eigenvectors, i.e. vectors $z^{(w,1)},\ldots,z^{(w,n_w)}$, such that $z^{(w,l)}\widehat \Lambda(x_w)=0, 1 \leq l \leq n_w$. The geometric multiplicity $n_w$
of eigenvalue $x_w$ is upper-bounded by its algebraic multiplicity $r_w$, which is defined as the multiplicity of root $x_w$ of $\det \widehat \Lambda(x)$.
Equation \eqref{mBasisSystemEval} is solvable if for  each $l$    $\sum_{j=0}^{u+v}z_j^{(w,l)}\lambda_{ji}(x_w) \neq 0$ for at least 
one $i:u+v<i\leq u+v+\delta$, i.e. if $\Lambda(x_w)$ is a full-rank matrix. The total number of such
pairs $(x_w,z^{(w,l)})$ is upper-bounded by  $N=\deg \det \widehat \Lambda(x)=\deg \det \mathcal Q(x)-\deg \det \mathcal B(x)$. 
The polynomials $Q_i(x,y),0 \leq i \leq u+v$ represent a Gr\"obner basis of some submodule of $M_{r_1+r_2,u+v+1}$, and could  be obtained from those given by $\mathcal B(x)$ 
by executing lines 6--10 of IIA for a few additional points $(x_i,y_i)$ and/or pairs $(\alpha,\beta)$. Hence, by lemma \ref{lDegDet},
  $\Delta_1=\deg \det \mathcal Q(x)=\Delta((Q_0(x,y),\ldots,Q_{u+v}(x,y)))$ and $\Delta_0=\deg \det \mathcal B(x)=n\frac{(r_1+r_2)(r_1+r_2+1)}{2}$.

Let polynomials $R_i(x,y),0 \leq i \leq w$ be a Gr\"obner basis with respect
to $(1,k-1)$-weighted degree lexicographic ordering of $I_R$ and $M_{R,w+1}$ for some $R$ and $w$. Then $\lt R_i(x,y)=x^{r_i}y^i, 0 \leq i \leq w$, where\footnote{There 
is no formal proof for this approximation. However, one can argue  that the polynomials in a Gr\"obner basis of $M_{R,w+1}$ should have {\em approximately} 
the same $(1,k-1)$-weighted degree, since the IIA, which can be used to construct them, always increases the degree of the smallest polynomial. Numerical
 experiments confirm this claim. Alternatively, if the received sequence is not a codeword,
a Gr\"obner basis of zero-dimensional ideal $I_R$ must contain the polynomials with $(1,k-1)$-weighted degree both
 below and above the value given by \eqref{mWDegBound}, and the approximate expression for $l_R$ derived below coincides with that one. }  $r_i\approx l_R-i(k-1), 0 \leq i \leq w-1$ for some $l_R$,
 $r_w=0$, and $\sum_{i=0}^w r_i=n\frac{R(R+1)}{2}$. Hence 
$wl_R -(k-1)\frac{w(w-1)}{2}\approx n\frac{R(R+1)}{2}$, and $l_R=l_R(w)\approx \frac{(k-1)w(w-1)+nR(R+1)}{2w}$.  

Since the polynomials $P_i(x,y)$ and $S_j(x,y)$ represent Gr\"obner bases of $M_{r_1,u+1}$ and $M_{r_2,v+1}$, $\lt P_i(x,y)=x^{p_i}y_i$, $\lt S_j(x,y)=x^{s_j}y^j$, where
$p_i\approx l_{r_1}(u)-i(k-1)$ and $s_j\approx l_{r_2}(v)-j(k-1)$. Then $t_i\approx l_{r_1}(u)+l_{r_2}(v)-i(k-1), 0 \leq i \leq u+v$, $t_{u+v}=0$ and 
$N=\Delta_1-\Delta_0=\sum_{i=0}^{u+v}t_i-n\frac{(r_1+r_2)(r_1+r_2+1)}{2}\approx (l_{r_1}(u)+l_{r_2}(v))(u+v)-(k-1)\frac{(u+v)(u+v-1)}{2}-n\frac{(r_1+r_2)(r_1+r_2+1)}{2}$
%\approx \left(\frac{(k-1)u(u-1)+nr_1(r_1+1)}{2u}+\frac{(k-1)v(v-1)+nr_2(r_2+1)}{2v}\right)(u+v)-(k-1)\frac{(u+v)(u+v-1)}{2}-n\frac{(r_1+r_2)(r_1+r_2+1)}{2}$. 
Hence,
\begin{equation}
\label{mDetDefect}
N\approx -(k-1)\frac{u+v}{2}+\frac{n}{2uv}\left((vr_1-ur_2)^2+v^2r_1+u^2r_2\right).
%nr'-u(k-1)<nr'-(k-1)\frac{-1+\sqrt{1+\frac{4nr'(r'+1)}{k-1}}}{2},
\end{equation}
% where the last inequality was obtained from \eqref{mRhoBound}.

Let us assume that the elements of $\Lambda(x)$ are univariate polynomials with independent coefficients uniformly distributed over $\F$. Then 
$\overline \lambda_{iwl}=\sum_{j=0}^{u+v}z_j^{(w,l)}\lambda_{ji}(x_w)$ is a random variable uniformly distributed over $\F^s$, where $\F^s$ is the smallest algebraic extension of $\F$,
such that $x_w\in \F^s$, and $s$ is the extension degree. 
Then the probability of $\overline \lambda_{iwl}$ being non-zero for at least one $i\in \{u+v+1,\ldots,u+v+\delta\}$ is given by $\theta_{s\delta}=1-\frac{1}{|F|^{s\delta}}.$

Consider factorization $\det \Lambda(x)=\alpha\prod_i \phi_i(x)$, where $\alpha\in \F\setminus\{0\}$, and $\phi_i(x)\in \F[x]$ are some monic irreducible polynomials.
Each eigenvalue $x_w$ is a root of at least one of $\phi_i(x)$, so $x_w\in \F^{\sigma_i}, \sigma_i=\deg \phi_i(x)$, and $N=\sum_i \sigma_i$. Let $\omega_j=|\{i|\sigma_i=j\}|, 1 \leq j \leq N$.
Observe that $\phi_i(x)$ has $\sigma_i$ distinct roots in $\F^{\sigma_i}$.
Assuming the worst case, where the geometric and algebraic multiplicities of eigenvalues are the same, one obtains the following expression
for the probability of \eqref{mBasisSystemEval} being solvable for all eigenvalues $x_w$:
\begin{equation}
\label{mSpecConvergProb}
\theta_\delta(\omega)=\prod_{j=1}^N\theta_{j\delta}^{j\omega_j}=\prod_{j=1}^N\left(1-\frac{1}{|\F|^{j\delta}}\right)^{j\omega_j}.
\end{equation}
Assuming that $\det \widehat\Lambda(x)$ is a polynomial with independent coefficients uniformly distributed over $\F$, one
can estimate the probability of obtaining a particular factorization of $\det \Lambda(x)$ as $P_\omega=\frac{1}{|\F|^N}\prod_{j=1}^N{{\mu_j+\omega_j+1}\choose \omega_j}$  \cite{dixon2004degree},
where $\mu_j$ is the number of monic irreducible polynomials of degree $j$. Hence, the probability of \eqref{mPseudoInverse} being solvable is given by 
$$\Theta(\delta)=\sum_\omega \theta_\delta(\omega)P_\omega,$$
where summation is performed over all partitions $\omega$ of $N$.

Exact evaluation of this expression does not seem to be feasible. However, it can be seen that the value of 
\eqref{mSpecConvergProb} is dominated by the first multiple, and it is known that a random polynomial over a finite field $\F$ has in 
average one root in it \cite{leontev2006roots}. Hence, the probability of \eqref{mPseudoInverse} being unsolvable decreases
exponentially fast with $\delta$. Thus, for sufficiently large $\F$ one can assume that  a Gr\"obner basis of $I_{r_1+r_2}$ can be derived from $u+v+1+\delta, \delta=O(1)$ polynomials 
given by \eqref{mRandomLinearComb}. 

The above analysis was performed for an impractical version of the proposed randomized ideal multiplication
method. It turns out that the polynomial matrix corresponding to 
the actual polynomials $Q_{0}(x,y),\ldots,Q_{u+v}(x,y)$ generated on line 2 of the $Merge$ algorithm has usually more
 than one eigenvalue in $\F$ with high algebraic multiplicity. But the geometric multiplicity
of the corresponding eigenvectors appears to be much less than the algebraic one (although still greater than $1$), so the algorithm still quickly converges.

Let us now estimate the number of iterations of $Reduce$ algorithm called on line 8 of $Merge$. To do this observe
that the objective of $Reduce$ is to decrease $(1,k-1)$-weighted degrees of polynomials constructed on lines 2 and 7 of $Merge$
from approximately $l_{r_1}+l_{r_2}$ to approximately $l_{r_1+r_2}$, i.e. to cancel the monomials with too high $(1,k-1)$-weighted degree.
For each polynomial approximately $(l_{r_1}+l_{r_2}-l_{r_1+r_2})(u+v)$ monomials should be eliminated. The total number of monomials to be eliminated can be estimated
as\footnote{Observe that the objective of minimization at line 2 of $Merge$ is to decrease the number
of monomials to be cancelled, i.e. decrease the number of iterations in $Reduce$.} $\sum_{i=0}^{u+v} (l_{r_1}+l_{r_2}-l_{r_1+r_2})(u+v)= \sum_{i=0}^{u+v} (l_{r_1}+l_{r_2}-i(k-1)-(l_{r_1+r_2}-i(k-1)))(u+v)=N(u+v)$. At least one monomial
 is cancelled during each iteration of $Reduce$. Taking into account \eqref{mDetDefect}, one obtains that the number of
iterations in $Reduce$ is given by $O(\tilde r(n-\sqrt{nk}))$, where $\tilde r=r_1+r_2$.   The algorithm operates with polynomials 
containing $O(n(2\tilde r)^2)$ terms, i.e. its complexity is given by $O(n(n-\sqrt{nk})\tilde r^3)$. 

It can be seen from \eqref{mRhoBound} that the number of polynomials in the basis of $I_{\tilde r}, \tilde r\leq r$ is $O(\tilde r\sqrt{n/k})$. 
The degrees of these polynomials
can be estimated as $\Wdeg_{(0,1)}Q_i(x,y)=O(\tilde r\sqrt{n/k})$ and $\Wdeg_{(1,0)}Q_i(x,y)=O(n\tilde r)$. Computing a product of two such polynomials
requires $O(n\tilde r^2\sqrt{n/k}\log(\tilde r\sqrt{n/k})\log(n\tilde r))$ operations. The analysis given above suggests
that the number of iterations performed by $Merge$ is $O(1)$. 
Therefore, the complexity of polynomial multiplications needed to  construct
 the Gr\"obner basis of $I_{2\tilde r}$ from the basis of $I_{\tilde r}$ is $O(\frac{n^2}{k}\tilde r^3\log(\tilde r\sqrt{n/k})\log (n\tilde r))$. 
Hence, one call to $Merge$ at line 11 of the interpolation algorithm requires $O(n\tilde r^3(a\log(\tilde r\sqrt{n/k})\log (n\tilde r)+b(n-\sqrt{nk})))$ operations for some 
positive $a$ and $b$. 

Obviously, the complexity of $Interpolate$ algorithm is dominated by the FOR loop.
The number of calls to  $Merge$ in this loop is given by
\begin{equation}
\label{mNumOfIterations}
M=\lfloor \log_2 r\rfloor +\sum_{i=0}^{\lfloor\log_2 r\rfloor} r_i.
\end{equation}
The second term in this expression corresponds to line 17 of the algorithm. 
The complexity of the whole algorithm is dominated by the
last iteration, so the overall complexity is given by $O(nr^3(a\log(r\sqrt{n/k})\log (nr)+b(n-\sqrt{nk}))$. Observe that this is better than the complexity
of IIA. 

\subsection{Re-encoding}
The proposed binary interpolation algorithm can be integrated with the re-encoding approach \cite{koetter2003complexity, koetter2003efficient,
ma2007complexity}.
As it was shown in section \ref{sBasisChange}, $M_{r,\rho}=[(y-T(x))^j\phi^{r-j}(x),y^s(y-T(x))^r,0 \leq j \leq r,1 \leq s \leq \rho-r]$. Let  $\psi(x)=\prod_{i=1}^k (x-x_i)$. 
Dividing $T(x)$ by  $\psi(x)$, one obtains $$T(x)=h(x)\psi(x)+g(x),$$ 
where $g(x_i)=y_i, 1 \leq i \leq k$ and $h(x_i)=\frac{y_i-g(x_i)}{\psi(x_i)}, i=k+1,\ldots,n$. Substituting $y=g(x)+z\psi(x)$  and dividing\footnote{This operation prevents one from using the concept of ideal here.} all polynomials in $M_{r,\rho}$
by $\psi^r(x)$, one obtains the module 
\begin{eqnarray*}
\widehat M_{r,\rho}&=&\left\{P(x,z)\in \F[x,y]\left|P\left(x_i,\frac{y_i-g(x_i)}{\psi(x_i)}\right)=0^r,\right.\right.\\
&&\left.i=k+1,\ldots,n,\Wdeg_{(0,1)}P(x,z)<\rho\right\},
\end{eqnarray*}
which is generated by
$\widehat \Pi_{r,j}(x,z)=(z-h(x))^j\theta^{r-j}(x),0 \leq j \leq r$ and $\widehat \Pi_{r,r+j}=z^j\psi^j(x)(z-h(x))^r, 1 \leq j \leq \rho-r$, where $\theta(x)=\frac{\phi(x)}{\psi(x)}$. 
There is a one-to-one correspondence between the polynomials in $M_{r,\rho}$ and $\widehat M_{r,\rho}$, and the smallest polynomial 
with respect to $(1,k-1)$-weighted
degree lexicographic  ordering in $M_{r,\rho}$ corresponds to the smallest polynomial with respect to $(1,-1)$-weighted degree lexicographic ordering in $\widehat M_{r,\rho}$. If a polynomial in $M_{r,\rho}$ has leading term $ax^iy^j$, then the corresponding polynomial in $\widehat M_{r,\rho}$ has leading
term $ax^{i+jk-rk}z^j, a\in \F$. This transformation essentially reduces the number of interpolation points. For high-rate codes this significantly decreases
the number of terms in the polynomials, reducing thus the overall algorithm complexity. 

The Gr\"obner basis of $\widehat M_{r,\rho}$ can be again constructed by the $Interpolate$ algorithm after minor modifications, as shown  in Figure 
\ref{fReencodedInterpolation}. $(1,-1)$-weighted degree lexicographic ordering must be used throughout this algorithm. 
\begin{figure}
\begin{algorithm}{ReencodeInterpolate}{((x_i,y_i),1 \leq i \leq n),r,k}
\psi(x)\=\prod_{i=1}^k (x-x_i), \theta(x)=\prod_{i=k+1}^n(x-x_i)\\
T(x)\=\sum_{i=1}^n y_i\frac{\prod_{j\neq i}(x-x_j)}{\prod_{j\neq i}(x_i-x_j)}\\
\text{Compute $h(x): T(x)=h(x)\psi(x)+g(x), \deg g(x)<k$}\\
\mathcal G\=(\theta(x))\\
j\=0\\
\begin{REPEAT}
\,{\mathcal G}\=\CALL {Reduce}(\mathcal G,(\psi(x)z)^j(z-h(x)))\\
j\=j+1
\end{REPEAT}{\lt G_j=  x^{(j-1)k}z^j}\\
\mathcal B\=\mathcal G\\
\text{Let $r=\sum_{j=0}^m r_j2^j, r_j\in \{0,1\}$}\\
R\=1\\
\begin{FOR}{j\=m-1 \TO 0 }
R\=2R\\
\mathcal B\=\CALL {Merge}(\mathcal B,\mathcal B,\delta(|\mathcal B|,|\mathcal B|,R))\\
\begin{IF}{r_j=1}
R\=R+1\\
\mathcal B\=\CALL {Merge}(\mathcal B,\mathcal G,\delta(|\mathcal B|,|\mathcal G|,R))
\end{IF}
\end{FOR}\\
\RETURN \mathcal B
\end{algorithm}
\caption{Construction of a Gr\"obner basis for $\widehat M_{r,\rho}$}
\label{fReencodedInterpolation}
\end{figure}
The algorithm first constructs a Gr\"obner basis of $\widehat M_{1,j+1}$. The corresponding loop
in the original algorithm terminates as soon as a polynomial with leading term $y^j$ is discovered. After change of variables
and term ordering the termination condition transforms to $LT G_j=x^{(j-1)k}z^j$. Then the algorithm proceeds with increasing
root multiplicity in the same way as the original algorithm. However, the termination threshold $\Delta_0$
of $Merge$ algorithm has  to be changed. For root multiplicity $R$ 
after change of variables  the basis polynomials
should have leading terms $x^{\widehat t_i}z^i, 0 \leq i \leq u-1$,  such that
$\widehat t_i=t_i+(i-R)k$, where $x^{t_i}y^i$ are the leading
terms of the corresponding polynomials which would be obtained without re-encoding.  Therefore, the termination threshold in the modified algorithm
should be set to $\widehat \Delta_0=\sum_{i=0}^{u-1} \left(t_i+(i-R)k\right)=\Delta_0-Rku+k\frac{(u-1)u}{2}=\Delta_0-k\frac{R(R+1)}{2}+
k\frac{(u-1-R)(u-R)}{2}$, where $\Delta_0$ is the termination threshold derived from  \eqref{mLTSum}. 
If the sizes of the bases supplied to $Merge$ are $u_1$ and $u_2$, then $u=u_1+u_2-1$.
Hence, one can compute the termination threshold for $Merge$ as  $\delta(u_1,u_2,R)=(n-k)\frac{R(R+1)}{2}+k\frac{(u_1+u_2-2-R)(u_1+u_2-1-R)}{2}.$

\section{Numeric results}
\label{sNumRes}
This section presents simulation results illustrating the performance  of the proposed algorithm.
Karatsuba fast univariate polynomial multiplication algorithm \cite{KnuthArt2} was used at steps 2 and 7 in $Merge$ algorithm. 

\begin{figure}
\centering
\includegraphics[width=0.47\textwidth]{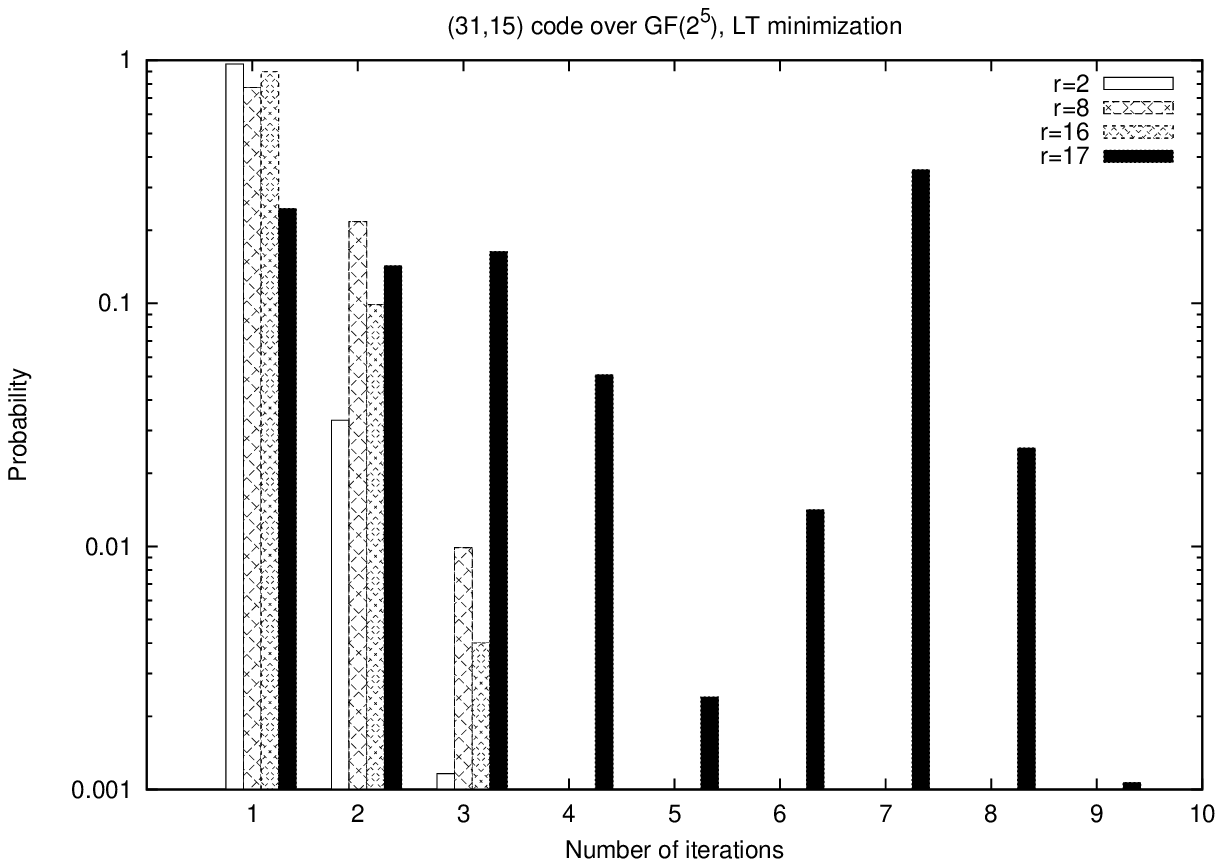}
\includegraphics[width=0.47\textwidth]{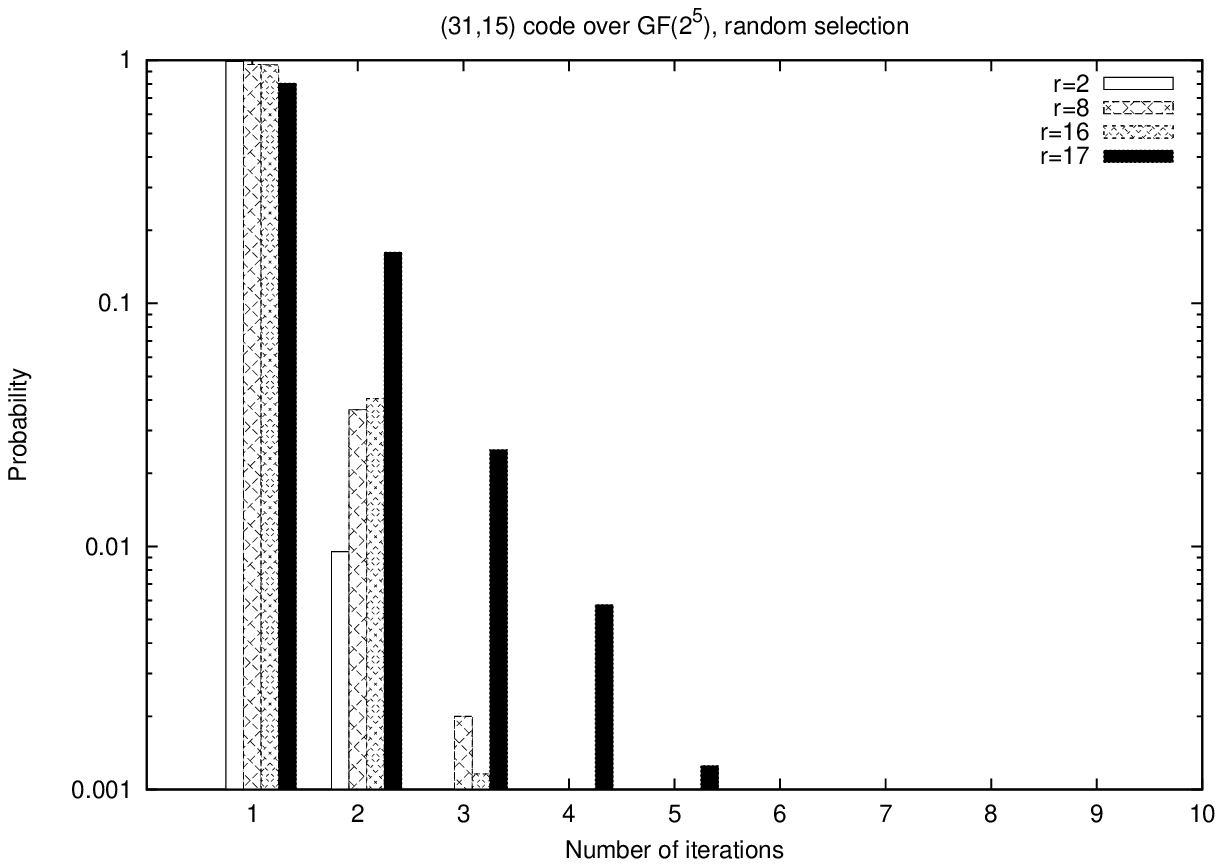}
\includegraphics[width=0.47\textwidth]{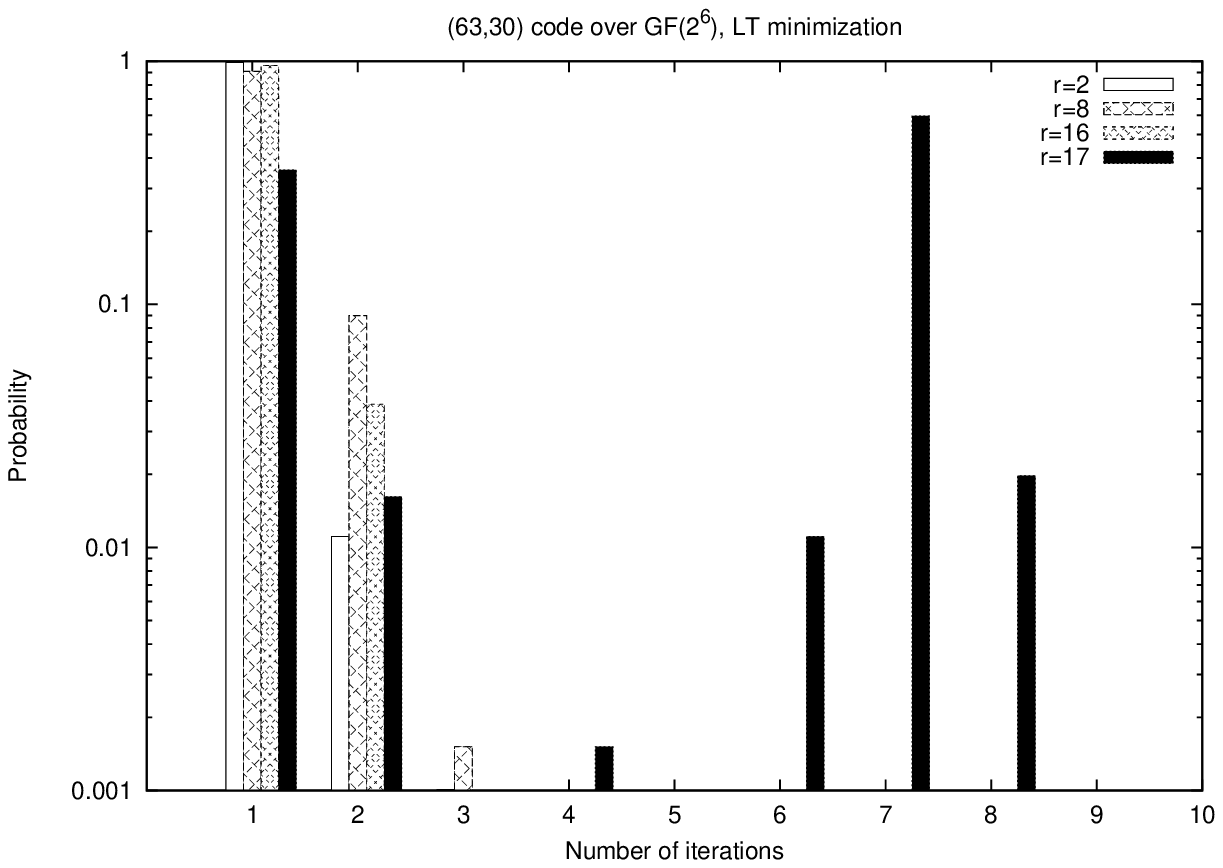}
\caption{Probability distribution of the number of iterations in $Merge$ algorithm}
\label{fMergeIt}
\end{figure}

Figure \ref{fMergeIt} illustrates the probability distribution of the number of iterations performed by $Merge$ algorithm while constructing Gr\"obner
bases of $I_r$ for different values of $r$.  $(31,15)$ code over $GF(2^5)$ and $(63,30)$ code over $GF(2^6)$ were considered.
First and third plots were obtained with the algorithm presented in Figure \ref{fFastModMul}. The second plot was obtained by replacing
leading term minimization on line 2 of the algorithm with random selection
of $Q_i(x,y)$ from the set $\{P_i(x,y)S_0(x,y),\ldots,P_{i-v}(x,y)S_v(x,y)\}$. It can be seen that in the latter  case the algorithm
indeed converges exponentially fast. In the first and third cases the convergence is still exponential for $r=2^m$, although a bit slower
than in the case of random polynomial selection. However, for the case of $r=17$ up to 8 iterations may be needed with high probability.
The reason is that constructing a Gr\"obner basis of $I_{17}=I_{16}\cdot I_1$ requires processing two different collections of polynomials $P_i(x,y), 0 \leq i \leq u$ and 
$S_j(x,y), 0 \leq j \leq v$ with $v<<u$. There is high probability that the smallest polynomial $S_{j_0}(x,y)$ is used almost for all $i$ on line $2$ of the $Merge$ algorithm.
This results in high algebraic and geometric  multiplicity of  $\widehat \Lambda(x)$ eigenvalues, i.e. the most probable partitions $\omega$ in \eqref{mSpecConvergProb}
are those with large $\omega_1$. This effect is compensated  by reduction of the total number of
 iterations in $Reduce$ algorithm.  Observe also that for $r=2^m$ the algorithm converges faster for the case of $|\F|=64$ compared to $|\F|=32$, as predicted by \eqref{mSpecConvergProb}.

\begin{figure}
\centering\includegraphics[width=0.47\textwidth]{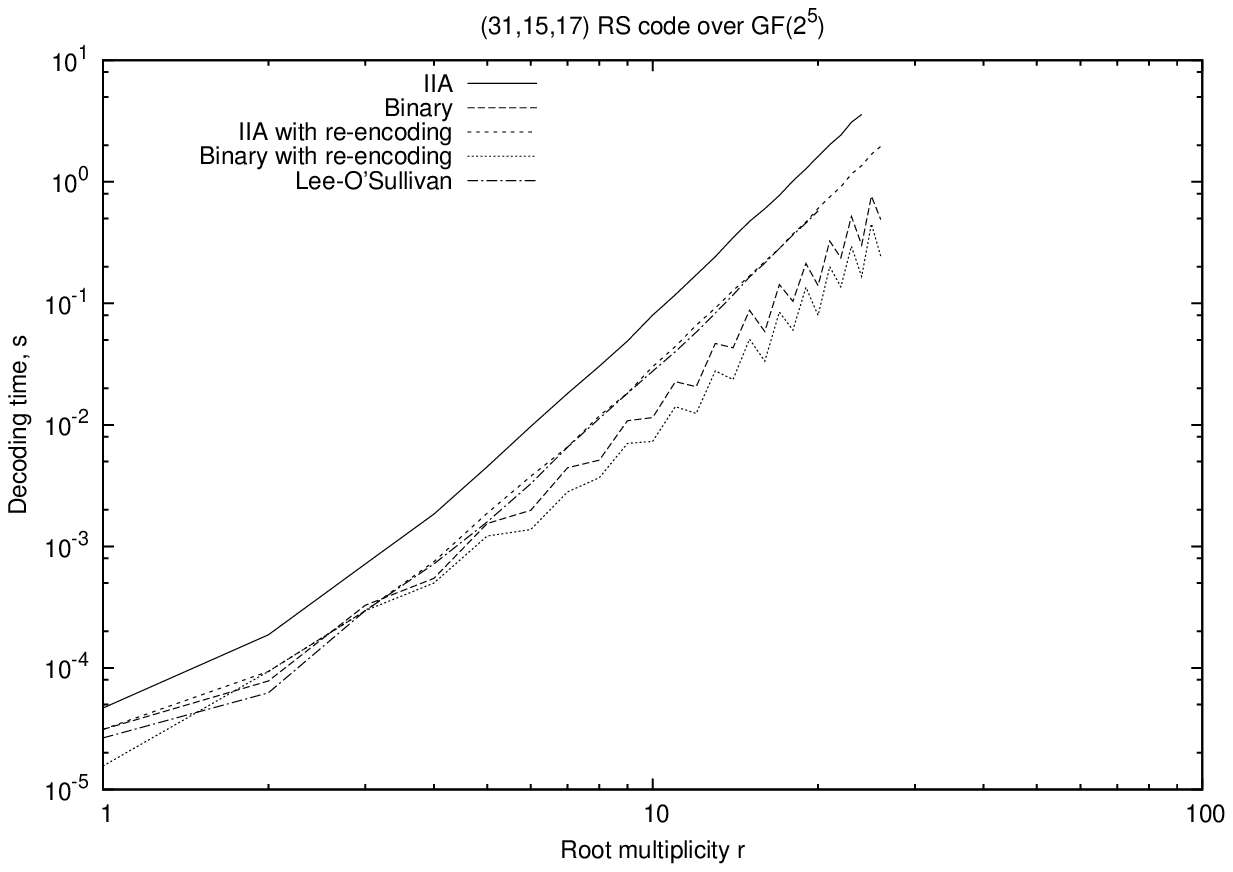}
\centering\includegraphics[width=0.47\textwidth]{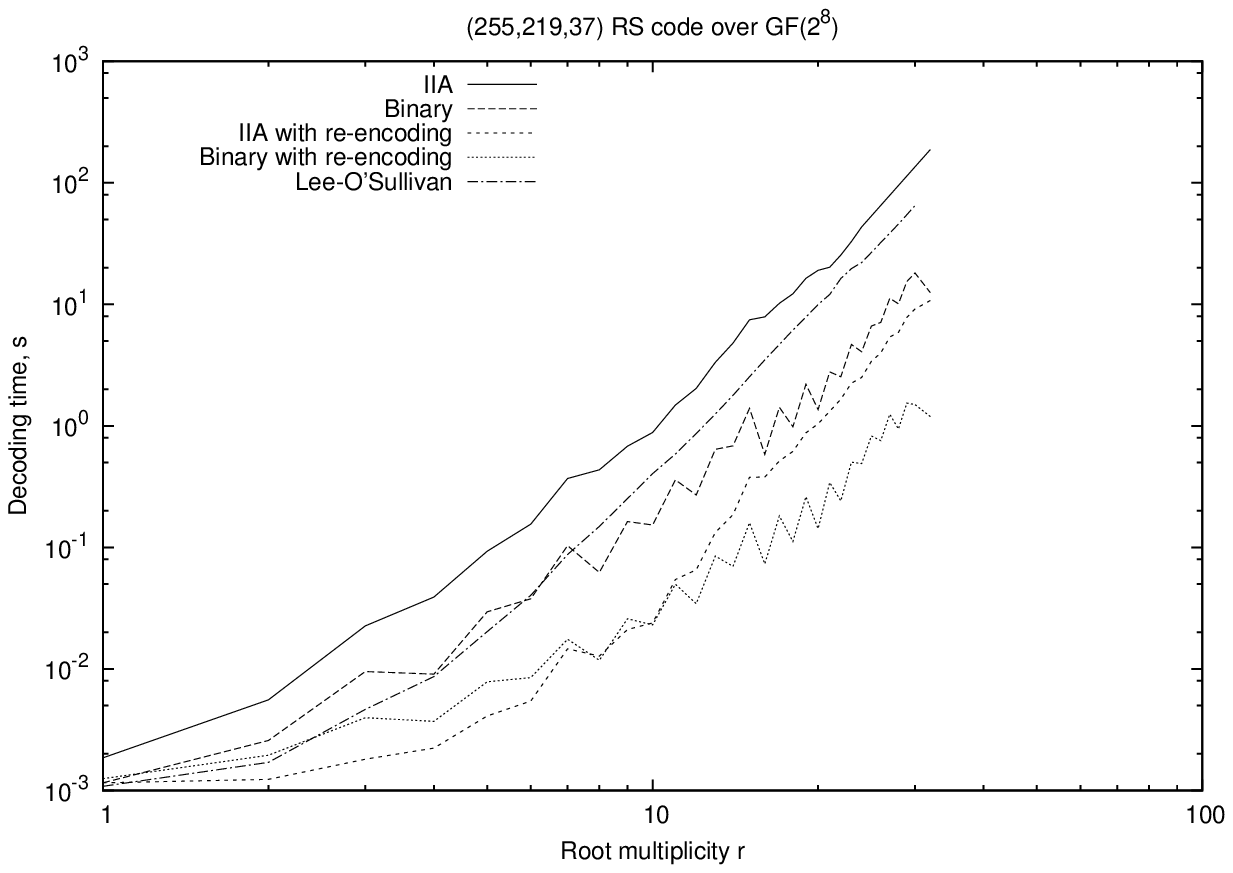}
\caption{Performance comparison of interpolation algorithms}
\label{fTiming}
\end{figure}
Figure \ref{fTiming} presents average list decoding time obtained with  IIA, Lee-O'Sullivan algorithm, 
proposed binary interpolation algorithm, re-encoding method based on IIA,
and binary interpolation algorithm with re-encoding.  $(31,15,17)$ and $(255,219,37)$ Reed-Solomon codes were considered. 
It can be seen that the proposed algorithm
provides up to 12  times lower complexity than IIA for the case of $(31,15,17)$ code, and up to 15 times lower complexity
for the case of $(255,219,37)$ code. Observe that the complexity of the proposed method increases
slower than for the case of  IIA, confirming thus the conclusion of Section  \ref{sComplexity}. The complexity of the Lee-O'Sullivan algorithm
turns out to be essentially the same as the one of IIA with re-encoding for rate-$1/2$ code, and exceeds it  considerably for high-rate code.
 Observe also, that 
in some cases increasing the root multiplicity reduces the complexity of the proposed interpolation method.
This represents the impact of the second term in \eqref{mNumOfIterations}, i.e. line 17 of the proposed algorithm. 

Observe also that the proposed algorithm outperforms the re-encoding method in the case of low-rate code. For 
the high-rate code the re-encoding method turns out to be better. However, employing re-encoding jointly with the proposed
method further reduces the complexity. The overall gain with respect to IIA is up to 22 times  for the case of $(31,15,17)$ code,
and up to 157 times for the case of $(255,219,37)$ code. 

\section{Conclusions}
In this paper a novel algorithm was proposed for the interpolation step of the Guruswami-Sudan list decoding algorithm.
It is based on the binary exponentiation algorithm, and can be considered as an extension of the Lee-O'Sullivan 
algorithm. The  proposed approach was shown to achieve significant asymptotical and practical gain compared to the case of iterative interpolation
algorithm. An important advantage of the new method is that its first step (first iteration of the WHILE loop in
 $Interpolate$ algorithm) coincides with  the Gao decoding algorithm, which is able to correct up to $(n-k)/2$ errors. Since the 
most likely error patterns can be corrected with this algorithm, one should invoke the remaining computationally expensive 
part of the proposed method only if the Gao algorithm does not produce a valid codeword. It is an open problem if
it is possible to terminate the interpolation algorithm as soon as it produces a bivariate polynomial
 containing all the solutions of a particular instance of the decoding problem, and avoid construction of $I_r$ basis 
for the worst-case $r$ given by \eqref{mRhoBound}-\eqref{mMinNumOfNonErrors}. Another interesting problem is to generalize
the proposed  algorithm to the case rational curve fitting problem considered in \cite{wu2008new}.

For the sake of simplicity, the proposed method was presented for the case of all interpolation points having the same multiplicity.
However, it can be extended to the case of weighted interpolation, allowing thus efficient implementation of soft-decision decoding.
Furthermore, it can be integrated with the re-encoding method, achieving thus  additional complexity reduction.

\section*{Acknowledgements}
The author thanks Dr. V.R. Sidorenko for many stimulating discussions. The author is indebted to
the  anonymous reviewers, whose comments have greatly improved the quality of the paper.

%\bibliographystyle{IEEEtran}
%\bibliography{coding,trifonov,math}

\end{document}